\documentclass{ecai}

\usepackage[english]{babel}
\usepackage{times}
\usepackage{soul}
\usepackage{url}
\usepackage[hidelinks]{hyperref}
\usepackage[utf8]{inputenc}
\usepackage[small]{caption}
\usepackage{graphicx, latexsym}
\usepackage{amsmath,thm-restate}
\usepackage[thmmarks, amsthm]{ntheorem}
\usepackage{booktabs}
\usepackage{algorithm}
\usepackage[capitalize]{cleveref}
\usepackage{apptools}
\usepackage{chngcntr}
\usepackage{subcaption}
\usepackage[switch]{lineno}

\usepackage{latexsym}

\renewtheorem{theorem}{Theorem}
\newtheorem{example}[theorem]{Example}
\newtheorem{lemma}[theorem]{Lemma}
\newtheorem{definition}[theorem]{Definition}
\newtheorem{corollary}[theorem]{Corollary}
\newtheorem{conjecture}[theorem]{Conjecture}
\newtheorem{remark}[theorem]{Remark}

\numberwithin{equation}{theorem}

\AtAppendix{\renewtheorem{theorem}{Theorem}[section]}

\usepackage{nico}
\usepackage{xparse,microtype}

\usepackage{subcaption}

\usepackage{tikz}
\usepackage{pgfplots}
\pgfplotsset{compat=1.13}
\usepgfplotslibrary{external}
\usepgfplotslibrary{colorbrewer}
\tikzexternaldisable
\tikzset{state/.style={rond5,minimum size=0.65cm}}
\tikzset{sink/.style={diam,minimum size=0.65cm}}
\tikzset{target/.style={rect,minimum size=0.6cm}}

\newcommand{\PO}{Player~$1$\xspace}
\newcommand{\PT}{Player~$2$\xspace}
\newcommand{\PLi}{Player~$i$\xspace}

\newcommand{\zug}[1]{\langle #1  \rangle}
\newcommand{\stam}[1]{}
\renewcommand{\set}[1]{\{ #1  \}}

\newcommand{\buchi}{B{\"u}chi~}
\newcommand{\pari}{\texttt{Parity}\xspace}

\newcommand{\frpari}{\texttt{FrParity}\xspace}
\newcommand{\frbuchi}{\texttt{Fr\buchi}\xspace}

\newcommand{\thresh}{\texttt{Th}\xspace}

\def\TUG{\operatorname{TOW}}
\def\tug{\operatorname{tow}}

\NewDocumentCommand{\Ngh}{mg}{\texttt{N}_{\IfNoValueTF{#2}{}{#2}}(#1)}

\NewDocumentCommand{\ThPo}{mg}{\thresh^{\pari}_1(\IfNoValueTF{#2}{}{#2,\xspace} #1)}
\NewDocumentCommand{\ThPt}{mg}{\thresh^{\pari}_2(\IfNoValueTF{#2}{}{#2,\xspace} #1)}
\NewDocumentCommand{\ThfrP}{mgg}{\thresh^{\frpari}_{#1}\IfNoValueTF{#2}{}{(\IfNoValueTF{#3}{}{#3,\xspace} #2)}}
\NewDocumentCommand{\ThfrB}{mmg}{\thresh^{\frbuchi}_{#1}(\IfNoValueTF{#3}{}{#3, \xspace} #2)}

\NewDocumentCommand{\frP}{g}{\texttt{frP}\IfNoValueTF{#1}{}{(#1)}}
\NewDocumentCommand{\frB}{g}{\texttt{frB}\IfNoValueTF{#1}{}{(#1)}}
\NewDocumentCommand{\Sld}{mg}{\texttt{S}_{#1}\IfNoValueTF{#2}{}{(#2)}\xspace}
\NewDocumentCommand{\Sed}{mg}{\texttt{T}_{#1}\IfNoValueTF{#2}{}{(#2)}\xspace}
\NewDocumentCommand{\FixSld}{g}{\texttt{S}\IfNoValueTF{#1}{}{(#1)}\xspace}
\NewDocumentCommand{\FixSed}{g}{\texttt{T}\IfNoValueTF{#1}{}{(#1)}\xspace}
\NewDocumentCommand{\buildto}{mg}{\texttt{t}_{#1}^0\IfNoValueTF{#2}{}{(#2)}\xspace}
\NewDocumentCommand{\buildti}{mgg}{\texttt{t}_{#1}^{\IfNoValueTF{#3}{i}{#3}}\IfNoValueTF{#2}{}{(#2)}\xspace}

\DeclarePairedDelimiter\ceil{\lceil}{\rceil}
\DeclarePairedDelimiter\floor{\lfloor}{\rfloor}

\usepackage[T1]{fontenc}

\newcommand{\G}{{{\cal G}}}

\newcommand{\play}{\text{play}\xspace}

\newcommand{\Nat}{\mathbb{N}}

\newcommand{\race}[1]{\text{race}(#1)\xspace}

\NewDocumentCommand{\stepbudget}{mmm}{T_{#3}^{#2}(#1)}
\NewDocumentCommand{\budget}{mm}{T_{#2}(#1)}
\NewDocumentCommand{\gamestepbound}{mm}{\mathsf{Steps}_{#1}(#2)}
\NewDocumentCommand{\stepop}{mmmm}{\mathrm{step}_{#4}(#1, #2, #3)}
\NewDocumentCommand{\step}{mmm}{\mathrm{step}_{#3}(#1, #2)}

\AtBeginDocument{
\setlength{\abovedisplayskip}{0.5\abovedisplayskip}
\setlength{\abovedisplayshortskip}{0.5\abovedisplayshortskip}
\setlength{\belowdisplayskip}{0.5\belowdisplayskip}
\setlength{\belowdisplayshortskip}{0.5\belowdisplayshortskip}
}

\crefname{lemma}{Lem.}{Lemms.}
\crefname{theorem}{Thm.}{Thms.}
\crefname{corollary}{Cor.}{Cors.}
\crefname{equation}{Eq.}{Eqs.}
\crefname{figure}{Fig.}{Figs.}
\crefname{tabular}{Tab.}{Tabs.}

\begin{document}
	
\begin{frontmatter}

\title{Reachability Poorman Discrete-Bidding Games}

\author[A]{\fnms{Guy}~\snm{Avni}}
\author[B,C]{\fnms{Tobias}~\snm{Meggendorfer}}
\author[A]{\fnms{Suman}~\snm{Sadhukhan}}
\author[D]{\fnms{Josef}~\snm{Tkadlec}}
\author[C]{\fnms{Đorđe}~\snm{Žikelić}} 

\address[A]{University of Haifa}
\address[B]{Technical University of Munich}
\address[C]{Institute of Science and Technology Austria}
\address[D]{Harvard University}


\begin{abstract}
We consider {\em bidding games}, a class of two-player zero-sum {\em graph games}. The game proceeds as follows. Both players have bounded budgets. 
A token is placed on a vertex of a graph, in each turn the players simultaneously submit bids, and the higher bidder moves the token, where we break bidding ties in favor of \PO. 
\PO wins the game iff the token visits a designated target vertex. 
We consider, for the first time, {\em poorman discrete-bidding} in which the granularity of the bids is restricted and the higher bid is paid to the bank. 
Previous work either did not impose granularity restrictions or considered {\em Richman} bidding  (bids are paid to the opponent).
While the latter mechanisms are technically more accessible, the former is more appealing from a practical standpoint. Our study focuses on {\em threshold budgets}, which is the necessary and sufficient initial budget required for \PO to ensure winning against a given \PT budget. 
We first show existence of thresholds. 
In DAGs, we show that threshold budgets can be approximated with error bounds by thresholds under continuous-bidding and that they exhibit a periodic behavior. We identify closed-form solutions in special cases. We implement and experiment with an algorithm to find threshold budgets.
\end{abstract}
\end{frontmatter}

\section{Introduction}
Two-player zero-sum {\em graph games} are a fundamental model with numerous applications, e.g., in reactive synthesis~\cite{PR89} and multi-agent systems~\cite{AHK02}. 
A graph game is played on a finite directed graph as follows. A token is placed on a vertex, and the players move it throughout the graph. We consider {\em reachability} games in which \PO wins iff the token visits a designated target vertex. Traditional graph games are {\em turn-based}: the players alternate turns in moving the token. We consider {\em bidding games}~\cite{LLPU96,LLPSU99} in which an ``auction'' (bidding) determines which player moves the token in each turn. 

Several concrete bidding mechanisms have been defined.
In all mechanisms, both players have bounded budgets.
In each turn, both players simultaneously submit bids that do not exceed their budgets, and the higher bidder moves the token.
The mechanisms differ in three orthogonal properties. {\em Who pays:} In {\em first-price} bidding only the winner pays the bid, whereas in {\em all-pay} bidding both players pay their bids. {\em Who is the recipient:} In {\em Richman} bidding (named after David Richman) payments are made to the other player, 
 in {\em poorman} bidding payments are made to the ``bank'', i.e.\ the bid is lost. 
{\em Restrictions on bids:} In {\em continuous-bidding} no restrictions are imposed and bids can be arbitrarily small, whereas in {\em discrete-bidding} budgets and bids are restricted to be integers.

In this work, we study, for the first time, first-price poorman discrete-bidding games. 
This combination addresses two limitations of previously-studied models. First, most work on bidding games focused on continuous-bidding games, where a rich mathematical structure was identified in the form of an intriguing equivalence with a class of stochastic games called {\em random-turn games}~\cite{PSSW09}, in particular for infinite-duration games~\cite{AHC19,AHI18,AHZ21,AJZ21}. These results, however, rely on bidding strategies that prescribe arbitrarily small bids. Employing such strategies in practice is questionable -- after all, money is discrete.
Second, discrete-bidding games have only been studied under Richman bidding~\cite{DP10,AAH21,AS22}.
The advantage of Richman over poorman bidding is that, as a rule of thumb, the former is technically more accessible. In terms of modeling capabilities, however, while Richman bidding is confined to so called {\em scrip systems} that provide fairness using an internal currency, poorman bidding captures a wide range of settings since it coincides with the popular first-price auction.

The central quantity that we focus on is the {\em threshold budget} in a vertex, which is a necessary and sufficient budget for \PO to ensure winning the game. Formally, a {\em configuration} of a bidding game is a triple $\zug{v, B_1, B_2}$, where $v$ denotes the vertex on which the token is placed and $B_i$ is \PLi's budget, for $i \in \set{1,2}$. For an initial vertex $v$, we call a function $T_v: \Nat \rightarrow \Nat$ the threshold budgets at $v$ if for every configuration $c = \zug{v, B_1, B_2}$, \PO wins from $c$ if $B_1 \geq T_v(B_2)$ and loses from $c$ if $B_1 \leq T_v(B_2)-1$. We stress that we focus only on {\em pure} strategies.

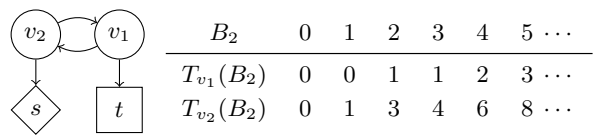
\begin{figure}[t]
	\centering
	\begin{minipage}{0.25\columnwidth}
		\centering
		\begin{tikzpicture}[auto,xscale=1.1]
			\draw node[sink] at (1,-1) (0) {$s$};
			\draw node[state] at (1,0) (1) {$v_2$};
			\draw node[state] at (2,0) (2) {$v_1$};
			\draw node[target] at (2,-1) (3) {$t$};
			
			\path[->]
				(1) edge (0)
				(1) edge[bend left] (2)
				(2) edge[bend left] (1)
				(2) edge (3)
			;
		\end{tikzpicture}
	\end{minipage} %
	\centering
	\begin{minipage}{0.7\columnwidth}
		\begin{tabular}{ccccccc}
			    $B_2$      & $0$ & $1$ & $2$ & $3$ & $4$ & $5\ \cdots$  \\
			\midrule
			$T_{v_1}(B_2)$ & $0$ & $0$ & $1$ & $1$ & $2$ & $3\ \cdots$ \\[2pt]
			$T_{v_2}(B_2)$ & $0$ & $1$ & $3$ & $4$ & $6$ & $8\ \cdots$ 
		\end{tabular}
	\end{minipage}
	\caption{\PO's threshold budget as a function of \PT's initial budget in the two intermediate vertices of the game on the left.}
\label{fig:TOW}
\end{figure}

\begin{example}\label{ex:TOW}
\normalfont
Consider the game that is depicted in Fig.~\ref{fig:TOW}, where we break bidding ties in favor of \PO.
In this example, we identify the first few thresholds. In \cref{thm:tug2}, we show that the thresholds in this game are $T_{v_1}(B_2) = \lfloor B_2/\phi \rfloor$ and $T_{v_2}(B_2) = \lfloor B_2 \cdot \phi \rfloor$, 
where $\phi \approx 1.618$ is the golden ratio.\footnote{We encourage the reader to read more about these two sequences in \url{https://oeis.org/A000201} and \url{https://oeis.org/A005206}. See also Remark~\ref{rem:Wythoff}.}
First, when both budgets are $0$, all biddings result in ties, which \PO wins and forces the game to $t$.
Second, we argue that \PO wins from $\zug{v_1, 0, 1}$. Indeed, \PO bids $0$. In order to avoid losing, \PT must bid $1$, wining the bidding and pays the bid to the bank. Thus, the next configuration is $\zug{v_2, 0,0}$, from which \PO wins.
Third, we show that $T_{v_2}(1) = 1$. Indeed, \PT wins from $\zug{v_2,0,1}$ by bidding $1$. On the other hand, from $\zug{v_2, 1, 1}$ \PO wins since by bidding $1$, he forces the game to $\zug{v_1,0,1}$, from which he wins.
Finally, $T_{v_1}(2) > 0$ since \PT can force two consecutive wins when the budgets are $\zug{0,2}$, and $T_{v_1}(2) = 1$ since by bidding $1$, \PO forces \PT to pay at least $2$ in order not to lose immediately, and he wins from $\zug{v_2,1,0}$. 
\end{example}

\paragraph{Applications.} 
In {\em sequential first-price auctions} $m$ items are sold sequentially in independent first-price auctions (e.g.,~\cite{LST12b,GS01}). The popularity of these auctions stems from their simplicity. Indeed, in each round of the auction, a user is only asked to bid for the current item on sale, whereas in {\em combinatorial auctions}, users need to provide an exponential input: a valuation for each subset of items.
Two-player sequential auctions are a special case of bidding games played on DAGs. Each vertex $v$ represents an auction for an item. A path from the root to $v$ represents the outcomes of previous rounds, i.e., a subset of items that \PO has purchased so far.  For a target bundle $T$ of items, this modeling allows us to obtain a bidding strategy that is guaranteed to purchase at least the bundle $T$ no matter how the opponent bids. Indeed, we solve the corresponding bidding game with the \PO objective of reaching a vertex in which $T$ is purchased.  We can also capture a quantitative setting in which \PO associates a value with each bundle of items. Given a target value $k$, we set \PO's target to be vertices that represent a purchased bundle of value at least $k$. We can then either find the threshold budget for obtaining value $k$ or fix the initial budgets and optimize over $k$.


Next, we describe two important classes of continuous poorman-bidding games that are technically challenging, and we argue that it is appealing to bypass this challenge by considering their discrete-bidding variants. Our study lays the basis for these extensions. 
First, all-pay poorman bidding games constitute a dynamic version of the well-known Colonel Blotto games \cite{Bor21}: we think of budgets as resources with no inherent value (e.g., time or energy) and a strategy invests the resources in order to achieve a goal. 
In fact, many applications of Colonel Blotto games are dynamic, thus all-pay bidding games are arguably a more accurate model~\cite{AIT20}. 
All-pay poorman bidding games are surprisingly technically complex, e.g., already in extremely simple games, optimal strategies rely on infinite-support distributions, and have never been studied under discrete bidding.
Second, the study of partial-observation bidding games was initiated recently~\cite{AJZ23}. Poorman bidding is both appealing from a theoretical and practical standpoint but is technically complex. Again, it is appealing to consider partial-information in combination with discrete bidding.

Finally, poorman discrete bidding are amenable to extensions such as multi-player games or non-zero-sum games~\cite{MKT18}.

\subsubsection*{Our Contribution}

\paragraph{Existence of thresholds.} 
In discrete-bidding games, one needs to explicitly state how bidding ties are resolved~\cite{AAH21}. Throughout the paper, we always break ties in favor of \PO. 
We start by showing existence of thresholds in every game, including games that are not DAGs. Our techniques are adapted from~\cite{AAH21} for Richman discrete-bidding games. We note that existence of thresholds coincides with {\em determinacy}: from every configuration, one of the players has a {\em pure} winning strategy. We point out that while determinacy holds in turn-based games for a wide range of objectives~\cite{Mar75}, determinacy of bidding games is not immediate due to the {\em concurrent} choice of bids. For example, {\em matching pennies} is a very simple concurrent game that is not determined: neither player can ensure winning.

\paragraph{Threshold budgets in DAGs.}
In continuous bidding, each vertex $v$ is associated with a {\em threshold ratio} which is a value $t \geq 0$ such that when the ratio between the two players' budgets is $t+\epsilon$, \PO wins, and when the ratio is $t-\epsilon$, \PT wins~\cite{LLPSU99}. 

First, we bound the discrete thresholds based on continuous ratios as follows. Let $t_v$ denote the continuous ratio at a vertex~$v$. Then, for every $B_2 \in \Nat$, we show that $T_v(B_2)$ lies in the {\em pipe}: $(B_2 - n) \cdot t_v \leq T_v(B_2) \leq B_2 \cdot t_v$, where $n$ is the number of vertices in the game.
We point out that the width of the pipe is fixed, so for large budgets $B_2$ the value $T_v(B_2)/B_2$ tends to the threshold ratio $t_v$.

Second, we show that threshold budgets in DAGs exhibit a periodic behavior. While we view this as a positive result, it has a negative angle: The periods are surprisingly complex even for fairly simple games, so even though we identify a compact representation for the thresholds in Example~\ref{ex:TOW}, we do not expect a compact representation in general games. 

Third, in continuous-bidding games, the compact representation of the thresholds (i.e., each vertex being associated with a ratio) is the key to obtaining a linear-time backwards-inductive algorithm to compute thresholds in DAGs. Under discrete bidding, given a \PT budget $B_2$, we present a pseudo-linear algorithm to find $T(B_2)$, namely its running time is linear in the size of the game and in $B_2$.

Fourth, we obtain closed-form solutions for a class of games called {\em race games}: for $a,b \in \Nat$, the race game $\race{a,b}$ ends within $a+b$ turns, \PO wins the game if he wins $a$ biddings before \PT wins $b$ biddings. For example, a ``best of 7'' tournament (as in the NBA playoffs) is $\race{4,4}$.

\stam{
{\bf Periodicity of the threshold function.} We show that every vertex $v$ in a game played on a DAG, there are values $u_x,u_y \in \Nat$ such that for a large enough budget $B_2$, we have $T(B_2 + u_x) = T(B_2) + u_y$. This interesting theoretical finding has a useful practical application: in order to compute $T(B_2)$, one needs to compute the thresholds for small initial budgets and extrapolate to larger initial budgets. (TODO: check and correct this. Do we need to add assumptions like different continuous thresholds?)

{\bf Closed-form solutions.} We identify closed-form solutions to some games. First, beyond the solution we describe in Example~\ref{ex:TOW}, we show a closed-form solution for a {\em tug-of-war} game (a game played on an undirected path) with three internal vertices, and show experimental evidence that closed-form solutions are unlikely for games beyond that. Second, for $a,b \in \Nat$, the {\em race game} $\race{a,b}$ (e.g.,~\cite{}) ends within $a+b$ turns, \PO wins the game if he wins $a$ biddings before \PT wins $b$ biddings. For example, a ``best out of 7'' tournament (as in the NBA playoffs) is $\race{4,4}$. We show that the threshold in $\race{a,b}$ is $....$.  
}
\begin{figure}[t]
\centering
\begin{tikzpicture}[auto]
		\draw (-4.225,0) node[white] {};

		\draw (0.3,0) node[rectangle,draw,align=center,minimum height=0.6cm] (0) {\small \textcolor{Dark2-C}{root}};
		\draw (-2,0) node[rectangle,draw,align=center,minimum height=0.6cm] (1) {\small {\textcolor{Dark2-A}{\(v_1\)}: $\race{4,5}$}};
		\draw (2.6,0) node[rectangle,draw,align=center,minimum height=0.6cm] (2) {\small {\textcolor{Dark2-B}{\(v_2\)}: $\race{3,5}$}};
		
		\draw (0) edge[->] (1);
		\draw (0) edge[->] (2);
\end{tikzpicture}%

\vspace{2ex}

\begin{tikzpicture}[auto]
	\begin{axis}[
			width=0.95\columnwidth,height=5cm,
			xmin=85.5,xmax=144.5,
			xlabel={\PT's budget $B_2$},ylabel={Threshold budget $\budget{B_2}{v}$},
			legend pos=north west,
			legend style = {mark options={scale=2}}
		]
		\addplot[mark=none, thin, lightgray, forget plot] coordinates {(0,0) (1000,800)};
		\addplot[mark=none, thin, lightgray, forget plot] coordinates {(0,-3.2) (1000,796.8)};

		\addplot[mark=none, thin, lightgray, forget plot] coordinates {(0,0) (1000,600)};
		\addplot[mark=none, thin, lightgray, forget plot] coordinates {(0,-2.4) (1000,597.6)};

		\addplot[mark=none, thin, lightgray, forget plot] coordinates {(0,0) (1000,711.111)};
		\addplot[mark=none, thin, lightgray, forget plot] coordinates {(0,-3.555) (1000,707.555)};

		\pgfplotsset{cycle list/Dark2-3}

		\addplot+[mark size=1pt,only marks,mark=star] table [x index=0, y index=2, col sep=comma] {data/fig2_race_choice.csv};
		\addlegendentry{$(4, 5)$}

		\addplot+[mark size=1pt,only marks,mark=diamond*] table [x index=0, y index=1, col sep=comma] {data/fig2_race_choice.csv};
		\addlegendentry{$(3, 5)$}

		\addplot+[mark size=1.75pt,only marks,mark=x] table [x index=0, y index=3, col sep=comma] {data/fig2_race_choice.csv};
		\addlegendentry{root}

		\addplot[mark=o,only marks,mark size=2.5pt,line width=1.25pt] coordinates {(44,28) (89,60) (134,92)};
	\end{axis}
\end{tikzpicture}

\caption{
	The thresholds in three vertices: a root vertex whose two children are roots of race games $\race{3,5}$ and $\race{4,5}$.
	For visibility, the x-axis starts at 85.
	We also depict the lower and upper bounds we obtain from our \emph{pipe theorem} (indicated by solid lines) and highlight two points indicating the periodicity in the root vertex.
}
\label{fig:intro-race}
\end{figure}
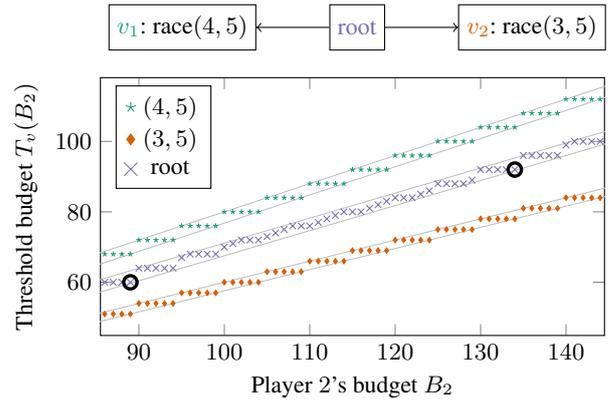

\begin{example}
\label{ex:2race+1}
\normalfont
We illustrate some of our main results. In Fig.~\ref{fig:intro-race}, we depict the threshold budgets in three vertices of a game 
as a function of \PT's budget. 
First, the discrete thresholds reside in a ``pipe'' 
with slope equal to the corresponding continuous ratio (\cref{stm:pipe}).  Second, $v_1$ and $v_2$ are roots of race games, thus their thresholds are simple step functions (\cref{thm:thresholdrace}).
Moreover, they lie exactly on the boundary of the pipe infinitely often, i.e.\ the pipe bound is tight (\cref{stm:pipe_tight}).
Third, the threshold budgets are periodic (\cref{thm:periodicity}), we have $T_r(B_2 + 45) = T_r(B_2) + 32$. We find it surprising that in such a simple game both the periodicity in the root node and the irregularity within this period are comparatively large.
\end{example}
\paragraph{Implementation and Experiments.} 
We provide a pseudo-polynomial algorithm to find the threshold budget given the initial budget of \PT in general games together with a specialized, faster variant for DAGs. We implement the algorithm, experiment with it, and develop conjectures based on our findings. Beyond the theoretical interest, the running time we observed is extremely fast, illustrating the practicality of finding exact thresholds. 




\section{Preliminaries}
	A reachability bidding game is \(\G  = \zug{V, E, t, s}\), where \(V\) is the set of vertices, \(E \subseteq V \times V\) is the set of edges, \PO's {\em target} is $t \in V$, a {\em sink} $s \in V$ has no path to $t$ and we think of $s$ as \PT's target, we assume that all other vertices have a path to both $t$ and $s$.
	We write $N(v) = \{u \mid (v, u) \in E\}$ to denote the \emph{neighbours} of $v$.
	
	A {\em configuration} of \(\calG\) is of the form \(c = \zug{v, B_1, B_2}\), where \(v \in V\) is the vertex on which the token is placed and $B_i$ is the budget of \PLi, for $i \in \set{1, 2}$. At $c$, both players simultaneously choose {\em actions}, and the pair of actions determines the next configuration. For $i \in \set{1,2}$, \PLi's action is a pair $\zug{b_i, u_i}$, where $b_i \leq B_i$ is an {\em integer} bid that does not exceed the available budget and $u_i \in N(v)$ is a neighbor of $v$ to move to upon winning the bidding. If $b_1 \geq b_2$, then \PO moves the token and pays ``the bank'', thus the next configuration is $\zug{u_1, B_1 - b_1, B_2}$. Dually, when $b_2 > b_1$, the next configuration is $\zug{u_2, B_1, B_2 - b_2}$. 
%
	
	A {\em strategy} is a function that maps each configuration to an action.\footnote{In full generality, strategies map {\em histories} of configurations to actions. However, {\em positional} strategies suffice for reachability games.} A pair of strategies $\sigma_1$, $\sigma_2$, and an initial configuration $c_0$ gives rise to a unique {\em play} denoted by $\play(c_0, \sigma_1, \sigma_2)$, which is defined inductively. The inductive step, namely the definition of how a configuration is updated given two actions from the strategies, is described above. Let $\play(c_0, \sigma_1, \sigma_2) = c_0, c_1, \ldots$, where $c_i = \zug{v_i, B^i_1, B^i_2}$. The {\em path} that corresponds to $\play(c_0, \sigma_1, \sigma_2)$ is $v_0,v_1,\ldots$

\begin{definition}[Winning Strategies]
A \PO strategy \(\sigma_1\) is called a \emph{winning strategy} from configuration \(c_0\) iff  for any \PT strategy \(\sigma_2\),  \(\play(c_0, \sigma_1, \sigma_2)\) visits the target $t$. 
On the other hand, a \PT strategy \(\sigma_2\) is a winning strategy from \(c_0\) iff for any \PO strategy \(\sigma_1\), \(\play(c_0, \sigma_1, \sigma_2)\) \emph{does not visit} the target \(t\).
For $i \in \set{1, 2}$, we say that \PLi \ {\em wins} from $c_0$ if he has a winning strategy from $c_0$. 
\end{definition}

Throughout the paper, we focus on the necessary and sufficient budget that \PO needs for winning, given a \PT budget, defined formally as follows. 

\begin{definition}[Threshold budgets]
Consider a vertex $v \in V$. The {\em threshold budget} at $v$ is a function $T_v: \Nat \rightarrow \Nat$ such that for every $B_2 \in \Nat$:
\begin{itemize}
\item \PO wins from $\zug{v, T_v(B_2), B_2}$, and 
\item \PT wins from $\zug{v, T_v(B_2) -1, B_2}$. 
\end{itemize}
\end{definition}

\section{Existence of Thresholds}\label{sec:determinacy}
In this section we show the existence of threshold budgets in games played on general graphs. 
\begin{definition}{\bf (Determinacy).} 
A game is {\em determined} if from every configuration, one of the players has a pure winning strategy.
\end{definition}

We claim that determinacy is equivalent to existence of thresholds. It is not hard to deduce both implications from the following observation. An additional budget cannot harm a player; namely, if \PO wins from a configuration $\zug{v, B_1, B_2}$, he also wins from $\zug{v, B'_1, B_2}$, for $B'_1 > B_1$, and dually for \PT. 

In the rest of this section, we prove determinacy of poorman discrete-bidding games. Our proof is based on a technique that was developed in \cite{AAH21} to show determinacy of \emph{Richman} discrete-bidding. We illustrate the key ideas. 
Consider a reachability bidding game \(\G = \zug{V, E, t, s}\) 
and a configuration \(c = \zug{v, B_1, B_2}\). 
We define a {\em bidding matrix} $M_c$ that corresponds to $c$. For $\zug{b_1, b_2} \in \set{0,\ldots, B_1} \times \set{0,\ldots, B_2}$, the \((b_1, b_2)^{\text{th}}\) entry in $M_c$ is associated with \PLi bidding $b_i$, for $i \in \set{1,2}$. We label entries in $M_c$ by $1$ or $2$ as follows. 
Let $\G_1$ denote a turn-based game that is the same as $\G$ only that in each turn, \PO reveals his bid first and \PT responds. Technically, once both players reveal their bids, the game proceeds to an intermediate vertex $i_{b_1, b_2} = \zug{b_1, b_2, c}$. Since $G_1$ is turn-based, it is determined, thus one of the players has a winning strategy from $i_{b_1, b_2}$. We label the \((b_1, b_2)^{\text{th}}\) entry in $M_c$ by $i \in \set{1,2}$ iff \PLi wins from $i_{b_1, b_2}$. For \(i \in \{1, 2\}\), we call a row or a column of \(M_c\) a \(i\)\emph{-row} or \(i\)\emph{-column}, respectively, if all its entries are labeled \(i\).


\begin{definition}{\bf (Local Determinacy)}
	A bidding game $\G$ is called \emph{locally determined} if for every configuration $c$, the bidding matrix $M_c$ either has a $1$-row or a $2$-column. 
\end{definition}

Local determinacy is used as follows. 
It can be shown that if \PO wins from $c$, then $M_c$ has a $1$-row. 
More importantly, suppose that \PO does not win in $c$, then local determinacy implies that there is a $2$-column, say $b_2$. This means that when \PT bids $b_2$ in $\G$, the game proceeds to a configuration $c'$ from which \PO does not win. 
In reachability games, since \PT's goal is to avoid the target, traversing non-losing configurations for \PT is in fact winning.\footnote{The theorem is stated for reachability objectives and it is extended in~\cite{AAH21} to richer objectives.}

\begin{lemma}{\bf (\cite[Theorem 3.5]{AAH21})}\label{lem:locallydetermined}
	If a reachability bidding game \(\G\) is locally determined, then \(\G\) is determined.
\end{lemma}

Local determinacy of poorman discrete-bidding games follows from the following observations on bidding matrices.  
\begin{restatable}{lemma}{observations}
	\label{lem:observations}
	Consider a poorman discrete-bidding game \(\G\) where \PO always wins tie, and consider a configuration \(c = \langle v, B_1, B_2 \rangle \).			(1) Entries in \(M_c\) in a column above the top-left to bottom-right diagonal are equal: for bids \(b_2 > b_1 > b_1'\), we have $M_c[b_1, b_2] = M_c[b_1', b_2]$. (2) Entries on a row, left of the diagonal are equal: for bids \(b_1 > b_2 > b_2'\), we have $M_c[b_1, b_2] = M_c[b_1, b_2']$. (3) The entry immediately under the diagonal equals the entry on the diagonal: For a bid $b$, we have $M_c[b, b]= M_c[b, b-1]$.
\end{restatable}

\begin{proof}
	If \(b_2 > b_1 > b_1'\) then \PT wins the current bidding for both pair of bids \(\zug{b_1, b_2}\) and \(\zug{b_1', b_2}\). 
	Thus \PT controls the corresponding intermediate vertex, and moves the token as per her choice. 
	As a result, only \PT's budget gets decreased by \(b_2\).
	Therefore, all the transitions that are available from \(\zug{c, b_1, b_2}\) are also available from \(\zug{c,b_1', b_2}\), and vice-versa.
	In other words, whoever wins from \(\zug{c, b_1, b_2}\) also wins from \(\zug{c, b_1', b_2}\), hence the entries are same. 
	The argument is similar when \(b_1 > b_2 > b_2'\).
	
	For the third observation, the tie-breaking mechanism, \PO always wins tie, plays the key role.
	For both the cases: when the bids are \(\zug{b, b}\), and when it is \(\zug{b, b -1}\) from a configuration \(c\), \PO wins the current bidding. 
	As a result the token moves according to his choice, his budget gets decreased by \(b\), while \PT's budget remains unchanged. 
	Therefore, all the available transitions from the \PO controlled vertex \(\zug{c, b, b}\) and \(\zug{c, b, b-1}\) are the same, and whoever wins from one, also wins from the other. 
	Thus the entries in \(M_c\) are same.
\end{proof}

%

The proof of \cite[Theorem 4.5]{AAH21} shows that a game whose bidding matrices have the properties of~\cref{lem:observations} is locally determined, irrespective of whether Richman or poorman bidding is employed. Combining with~\cref{lem:locallydetermined}, we obtain the following. 

\begin{restatable}{theorem}{determinacy}\label{thm:determinacy}
	Reachability poorman discrete-bidding games are determined.
\end{restatable}

\section{Threshold Budgets for Games on DAGs}\label{sec:pipeandstabilization}
In this section, we focus on games played on directed acyclic graphs (DAGs). We present two main results:
First, the \textit{Pipe theorem} that relates the threshold budgets to the threshold ratio in the continuous-bidding game; and,
second, the \textit{Periodicity theorem} which shows that the threshold budgets eventually exhibit a periodic behavior.
Throughout this section, let \(\calG = \zug{V, E, t, s}\) be a game with $\zug{V,E}$ a DAG.

\subsection{Relating Discrete and Continuous Thresholds}
We call the following theorem the \textit{Pipe theorem} since it shows that the threshold budgets $T_v(B_2)$ lie in a ``pipe'' below a line whose slope is the threshold ratio $t_v$ (see Example~\ref{ex:2race+1}). We note that threshold ratios can be computed in DAGs in time polynomial in the size of the game (a fact we also exploit later on in our algorithm on DAGs), thus an immediate corollary of the Pipe theorem is an efficient approximation algorithm to computing the threshold budgets. In Corollary~\ref{stm:pipe_tight}, we show that the lower bound is tight. For a vertex $v$, let {\em $\textnormal{max-path}(v)$} denote the length of the longest path from $v$ to either $t$ or $s$. Note that $\textnormal{max-path}(v) \leq |V| - 1$. 

\begin{restatable}[Pipe theorem]{theorem}{pipe} \label{stm:pipe}
	Given $v\in V$, denote by $t_v$ the threshold ratio in the continuous-bidding game at $v$. Then, for every initial budget $B_2\in\mathbb{N}$ of \PT, we have
	\begin{equation*}
		t_v \cdot (1- \textnormal{max-path}(v) / B_2)  \leq T_v(B_2) / B_2 \leq t_v.
	\end{equation*}
	The right-hand side inequality holds even when $\G$ is not a DAG.
\end{restatable}

\begin{proof}
	{\em Right-hand-side inequality.} We first prove the right-hand-side inequality. To prove that $T_v(B_2) / B_2 \leq t_v$, it suffices to prove that, for each $\epsilon > 0$, \PO has a winning strategy if the game starts in $v$, \PO's initial budget is at least $t_v\cdot B_2 + \epsilon$ and \PT's initial budget is $B_2$. If we are able to prove this claim, it will then follow that $T_v(B_2) / B_2 \leq t_v + \epsilon$ holds for every $\epsilon>0$, therefore $T_v(B_2) / B_2 \leq t_v$.
	
	Fix $\epsilon > 0$. We construct the winning strategy of \PO as follows. By the definition of the continuous threshold $t_v$, we know that \PO has a winning strategy in the poorman {\em continuous-bidding} game. Moreover, it was shown in~\cite[Theorem 7]{LLPSU99} that \PO has a {\em memoryless} winning strategy, i.e.~a strategy in which the bids and token moves in each turn depend only on the position of the token and the players' budgets. We take such strategy $\sigma_{\textrm{cont}}$. We then construct a winning strategy $\sigma_{\textrm{disc}}$ of \PO in the poorman discrete-bidding game as follows:
	\begin{itemize}
		\item At each turn, if \PO under $\sigma_{\textrm{cont}}$ would bid $b$, then \PO under $\sigma_{\textrm{disc}}$ bids $\lfloor b \rfloor$.
		\item If \PO wins the bidding, then the token is moved to the vertex perscribed by $\sigma_{\textrm{cont}}$. 
	\end{itemize}
	
	We show that $\sigma_{\text{disc}}$ is indeed winning for \PO. To do this, we prove that $\sigma_{\textrm{disc}}$ preserves the invariant that, whenever the token is in some vertex $v'$, the ratio of players' budgets is positive and strictly greater than $t_{v'}$. This invariant implies that the token does not reach the sink state as the continuous threshold in the sink state is infinite. Thus, as a poorman discrete-bidding game ends in finitely many steps, this then implies that the game must eventually reach \PO's target state and therefore that $\sigma_{\text{disc}}$ is winning for \PO.
	
	We prove the invariant by the induction on the length of the game play. The base case holds by the assumption that, in the initial vertex $v$, \PO's initial budget is at least $t_v\cdot B_2 + \epsilon$ and \PT's initial budget is $B_2$. Now, for the induction hypothesis, suppose that the token is in vertex $v'$ after finitely many steps, with \PT's budget $B_2'$ and \PO's budget at least $t_{v'} \cdot B_2' + \epsilon'$ for some $\epsilon'>0$. We show that the invariant is preserved in the next step. Suppose that \PO under $\sigma_{\textrm{disc}}$ bids $\lfloor b \rfloor$ where $b$ is the bid of \PO under $\sigma_{\textrm{cont}}$. In what follows, we use the fact that $\sigma_{\textrm{cont}}$ preserves the ratio invariant which was established in the proof of~\cite[Theorem 7]{LLPSU99}. We distinguish between two cases:
	\begin{itemize}
		\item If \PO wins the bidding and moves the token to $v''$, then the ratio of budgets at the next step is
		\begin{equation*}
			\begin{split}
				\frac{t_{v'} \cdot B_2' + \epsilon' - \lfloor b\rfloor}{B_2'} &\geq \frac{t_{v'} \cdot B_2' + \epsilon' - b}{B_2'} \\
				&> \frac{t_{v'} \cdot B_2' - b}{B_2'} \\
				&\geq t_{v''}.
			\end{split}
		\end{equation*}
		The last inequality follows from the fact that the fraction in the second line is the subsequent ratio of budgets under the continuous-bidding winning strategy.
		\item If \PT wins the bidding, then \PT had to bid at least $\lfloor b \rfloor + 1$. Suppose that \PT moves the token to $v''$ upon winning. Then the ratio of budgets at the next step is at least
		\begin{equation*}
			\begin{split}
				\frac{t_{v'} \cdot B_2' + \epsilon'}{B_2' - \lfloor b\rfloor - 1 } &\geq \frac{t_{v'} \cdot B_2' + \epsilon'}{B_2' - b} \\
				&> \frac{t_{v'} \cdot B_2'}{B_2' - b} \\
				&\geq t_{v''}.
			\end{split}
		\end{equation*}
		The first inequality follow by observing that $\lfloor b \rfloor + 1 \geq  b$, and the third inequality follows from the fact that the second fraction is an upper bound on the subsequent ratio of budgets under the continuous-bidding winning strategy.
	\end{itemize}
	Hence, the invariant claim for poorman discrete-bidding follows by induction on the length of the game play, thus $\sigma_{\text{disc}}$ is indeed winning for \PO and the right-hand-side inequality in the theorem follows.
	
	\noindent {\em Left-hand-side inequality.} We now prove the left-hand-side inequality. If $B_2 < \textrm{max-path}(v) $, the claim trivially follows. Otherwise, it suffices to prove that \PT has a winning strategy if the game starts in $v$, \PO's initial budget is strictly less than $t_v\cdot (B_2 - \textrm{max-path}(v))$ and \PT's initial budget is $B_2$.
	
	Suppose that $B_2 \geq \textrm{max-path}(v)$ and let $B_1  < t_v\cdot (B_2 - \textrm{max-path}(v) )$ be the initial budget of \PO. We construct the winning strategy of \PT as follows. Let $\sigma_{\textrm{cont}}$ be the memoryless winning strategy of \PT under {\em continuous-bidding} when the game starts in $v$, \PT's initial budget is $B_2 - \textrm{max-path}(v)$ and \PO's initial budget is $B_1$. Since  $B_1  < t_v\cdot (B_2 - \textrm{max-path}(v))$, such a strategy exists by the definition of the continuous threshold $t_v$ and by~\cite[Theorem 7]{LLPSU99} which shows that it is possible to pick a memoryless winning strategy. We then construct a winning strategy $\sigma_{\textrm{disc}}$ of \PT in the poorman discrete-bidding game when \PT has initial budget $B_2$ and \PO has initial budget $B_1$ as follows:
	\begin{itemize}
		\item At each turn, if \PT under $\sigma_{\textrm{cont}}$ would bid $b$, then \PT under $\sigma_{\textrm{disc}}$ bids $\lceil b \rceil$.
		\item If \PT wins the bidding, then the token is moved to the vertex perscribed by $\sigma_{\textrm{cont}}$. 
	\end{itemize}
	Note that, if we show that the bids $\lceil b \rceil$ under $\sigma_{\textrm{disc}}$ are legal (i.e.~do not exceed available budget), then $\sigma_{\textrm{disc}}$ is clearly winning for \PT. Indeed, $\sigma_{\textrm{cont}}$ is winning for \PT, the bids of $\sigma_{\textrm{disc}}$ are always as least as big as those of $\sigma_{\textrm{cont}}$ and the token moves under two strategies coincide. So we only need to prove that the bids $\lceil b \rceil$ under $\sigma_{\textrm{disc}}$ are legal. But this follows from the fact that the underlying graph is a DAG and thus the game takes at most $\textrm{max-path}(v)$ turns before it reaches either the target or the sink vertex. Hence, as the bids are legal under $\sigma_{\textrm{cont}}$ when \PT has initial budget $B_2 - \textrm{max-path}(v)$, the bids are also legal under $\sigma_{\textrm{disc}}$ as \PT can bid $b + 1 \geq \lceil b \rceil$ in each turn. This concludes the proof of the left-hand-side of the inequality.
\end{proof}

An immediate corollary of Thm.~\ref{stm:pipe} is that the ratio $T_v(B_2)/B_2$ tends to $t_v$. 

\begin{corollary}[Convergence to continuous ratios]\label{cor:convergence}
	For every $v \in V$ we have $\lim_{B_2\rightarrow\infty} T_v(B_2) / B_2 = t_v$.
\end{corollary}

\subsection{Periodicity of Threshold Budgets}
The following theorem shows that for any fixed $v\in V$ the function $T_v(\cdot)$ that yields the threshold budgets exhibits an eventually periodic behavior, as seen in Example~\ref{ex:2race+1}. 

\begin{restatable}[Periodicity theorem]{theorem}{periodicity} \label{thm:periodicity}
For any vertex $v\in V$ there exist values $B, u_x,u_y \in\Nat$ such that for all $B_2\ge B$ we have $T_v(B_2+u_x)=T_v(B_2)+u_y$.
Moreover, the values $B$, $u_x$, $u_y$ can be computed in polynomial time.
\end{restatable}

\begin{proof}
	We proceed by induction with respect to the topological order of the graph.
	For the target $v$ we set $(B,u_x,u_y)=(0,1,0)$ and for the sink we set $(B,u_x,u_y)=(0,0,1)$.
	Next, given a non-leaf vertex $v$, suppose that among its children there are $k$ distinct continuous threshold ratios, and denote them by $t_1<t_2<\dots<t_k$.
	Note that whenever \PO wins a bidding at vertex $v$ (while \PT has budget $B$), he moves the token to a child $u\in N(v)$ of $v$ with minimal value $T_u(B)$.
	We claim that for large enough $B$, the only relevant children $u$ are those with $t_{u}=t_1$.
	Indeed, consider two children $u,w\in N(v)$, one with $t_u=t_1$ and the other one with $t_w\ne t_1$.
	Then by \cref{stm:pipe}, for $B>t_2\cdot n / (t_2-t_1)$ we have
	\[T_w(B) \ge t_w\cdot (B-n) \ge t_2\cdot (B-n) > t_1\cdot B \ge T_u(B),
	\]
	thus \PO would prefer to move to $u$ rather than to $w$.
	
	Next, let $V^-=\{u\in N(v) \mid t_u=t_1\}$ be a set of those ``relevant'' children of $v$.
	We say that a function $f\colon \Nat\to\Nat$ is \textit{$(x,y)$-climbing} if it satisfies $f(B+x)=f(B)+y$ for all large enough $B$.
	By induction assumption, for each $u\in V^-$ the function $T_{u}(B)$ is $(u_x,u_y)$-climbing with a slope $u_y/u_x$. By \cref{stm:pipe}, this slope is equal to $t_1$.
	Thus, the function $T_u(B) - t_1\cdot B$ is periodic with period $u_x$.
	The function $\min_{u\in V^-}\{T_u(B) - t_1\cdot B\}$ is then periodic with the period equal to the least common multiple $l=\operatorname{lcm}\{u_x\mid u\in V^-\}$ of the respective periods.
	Therefore, the function $T_{v^-}(B) := \min_{u\in V^-}\{T_u(B)\}$ is $(l,l\cdot t_1)$-climbing.
	To summarize, the moves of \PO (upon winning a bidding) are faithfully represented by him moving the token to a vertex $v^-$ for which the threshold budgets are $(l,l\cdot t_1)$-climbing.
	Completely analogously we show that the moves of \PT are faithfully represented by her moving the token to a vertex $v^+$ with $(l',l'\cdot t_k)$-climbing threshold budgets.
	
	From now on, for ease of notation suppose that $T_{v^-}$ is $(u_x,u_y)$-climbing and that $T_{v^+}$ is $(w_x,w_y)$-climbing (for some integers $u_x,u_y,w_x,w_y$). We will show that $T_v$ is $(u_x\cdot(w_x+w_y), w_y\cdot(u_x+u_y))$-climbing. This will complete the induction proof.
	
	To prove this claim, it is convenient to represent each configuration $c=\zug{v,B_1,B_2}$ as a point in the plane with coordinates $[B_2,B_1]$, see~\cref{fig:u-climbing}.
	\PO can then force a win from a configuration $c=\zug{v,B_1,B_2}$ if and only if 
	the point $P=[B_2,B_1]$ lies on or above the threshold function~$T_v$.

	\begin{figure}[h!] 
		\centering
		\includegraphics[width=1\linewidth]{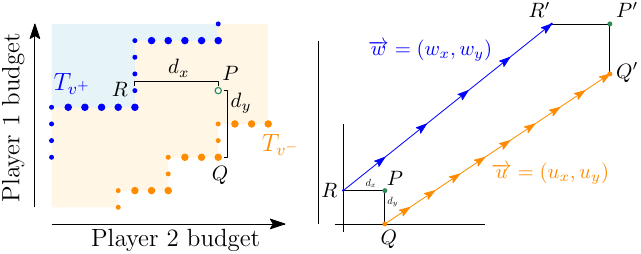}
		\caption{Left: Point $P$ lies on or above $T_v$ if and only if $d_x\le d_y+1$. Right: Chaining $v_x+v_y$ copies of $u$ and $u_x+u_y$ copies of $v$, the situation repeats.
		}
		\label{fig:u-climbing}
	\end{figure}

	Take any point $P=(P_x,P_y)$.
	Let $Q=(Q_x=P_x,Q_y)$ be the furthest point below $P$ that still lies on or above $T_{v^-}$, and
	let $R=(R_x,R_y=P_y)$ be the closest point to the left of $P$ that lies on or above $T_{v^+}$.
	Note that $P$ lies on or above $T_v$ if and only if the distances $d_y:=P_y-Q_y$ and $d_x:=P_x-R_x$ satisfy $d_x\le d_y+1$. Indeed, if the inequality holds then \PO can force a win by bidding $d_y$, whereas in the other case \PT can force a win by bidding $d_y+1$. 
	
	As a final step, we show that at some point further along the curves $T_{v^-}$ and $T_{v^+}$, the two distances $d_x$, $d_y$ increase by the same margin.
	Specifically, chain $u_x+u_y$ copies of a vector $(w_x,w_y)$ starting from point $R$ to get to point $R'=(R'_x,R'_y)$, and, similarly,
	chain $w_x+w_y$ copies of $(u_x,u_y)$ from $Q$ to $Q'=(Q'_x,Q'_y)$.
	Finally, let $P'$ be the point above $Q'$ and to the right of $R'$.
	Then a straightforward algebraic manipulation shows that distances from $P'$ to $Q'$ and to $R'$ both increased by $u_xw_y - w_xu_y$.
	Indeed, without loss of generality set $Q=(d_x,0)$ and $R=(0,d_y)$.
	Then we have
	\[Q'=(d_x+ (w_x+w_y)u_x, 0+(w_x+w_y)u_y)\]
	and
	\[R'=( (u_x+u_y)w_x, d_y+(u_x+u_y)w_y),\]
	so 
	\[P'=(Q'_x,R'_y)=(d_x+ (w_x+w_y)u_x, d_y+(u_x+u_y)w_y)\]
	and finally
	\begin{align*}
		P'_y-Q'_y&=d_y+(u_x+u_y)w_y - (w_x+w_y)u_y\\
		&= d_y + (u_xw_y - w_xu_y),\\
		P'_x-R'_x &= d_x+ (w_x+w_y)u_x - (u_x+u_y)w_x \\
		&= d_x + (w_yu_x - u_yw_x),
	\end{align*}
	concluding the proof. 
\end{proof}

This result implies that for each $v\in V$, the function $T_v(\cdot)$ can be finitely represented: let $B$ be \PT's budget when the period ``kicks in'', then for all $B' \leq B$, the value $T_v(B')$ is stored explicitly and these values can be extrapolated to find $T_v(B'')$ for $B'' > B$. 

We point out that periodicity may indeed appear only ``eventually'', as illustrated by \cref{fig:eventually_periodic}; namely, only at $B = 7$ state $(2, 2)$ continuously is an optimal choice and the periodic behaviour is observed.
Replacing $\race{5,4}$ with $\race{2x + 1, 2x}$ leads to quickly growing periodicity thresholds $B$.
Finally, we note that on non-DAGs, the behaviour is not necessarily periodic, as illustrated by \cref{thm:tug2} below.

\stam{
Let $v$ that is not a leaf, by using Thm.~\ref{stm:pipe} and the fact that a minimum of two periodic functions is itself periodic, we are able to reduce to the case in which there exists
a single descendant $v^-$ of $v$ where \PO always moves the token (upon winning a bidding), and
a single descendant $v^+$ of $v$ where \PT always moves the token.
Then we show that if $v^-$ satisfies $T_{v^-}(B_2+u^-_x)=T_{v^-}(B_2)+u^-_y$ and
$v^+$ satisfies $T_{v^+}(B_2+u^+_x)=T_{v^+}(B_2)+u^+_y$ (both for all large enough $B_2$),
then $v$ satisfies the same equality with $u_x=u^-_x\cdot (u^+_x+u^+_y)$ and $u_y=u^+_y\cdot(u^-_x+u^-_y)$. This completes the induction proof.
\end{proof}
As a consequence of Thm.~\ref{thm:periodicity}, we obtain that for each $v\in V$ the whole function $T_v(\cdot)$ can be represented using a finite amount of information -- its values before the period kicks in, and its values over one period.\todo{Sell this more; maybe "finite representation"}
Moreover, the periodic behaviour indeed appears only eventually, as illustrated by Figure~\ref{fig:eventually_periodic}.
There, it takes up to $B = 79$ until state $(2, 2)$ continuously is an optimal choice.
Only then do we observe periodic behaviour.
}

%
%
%
\begin{figure}[t]
	\centering
	\begin{tikzpicture}
	\begin{axis}[
			width=0.9\columnwidth,height=2.75cm,
			xlabel={\PT's budget $B_2$},
			ytick={1,2},yticklabels={{$(2,2)$}, {$(5, 4)$}},
			xmin=0.5,xmax=27.8,ymin=0.5,ymax=2.5,axis x line*=bottom,
			axis y line*=left,y tick label style={major tick length=0pt}
		]
		\addplot[only marks,mark color=black,mark size=2.2pt,mark=square*] table [x index=0, y index=2, col sep=comma] {data/winning_bids_race22_race1110.csv};
		\addplot[only marks,mark color=black,mark size=2.2pt,mark=square*] table [x index=0, y index=1, col sep=comma] {data/winning_bids_race22_race1110.csv};
	\end{axis}
	\end{tikzpicture}
	\caption{
We consider a game comprising a root node $v$ with two children, which are roots to \race{5,4} and \race{2,2}.
We depict \PO's {\em winning moves}: for each \PT's budget $B_2$, we depict the vertex (or vertices) that \PO may proceed to upon winning the bidding at configuration $\zug{v, T_v(B_2), B_2}$.
}
	 \label{fig:eventually_periodic}
\end{figure}
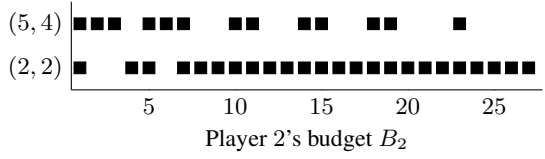

\section{Closed-form Solutions}\label{sec:closed-form}

In this section, we show closed-form solutions for threshold budgets in two special classes of games.

\subsection{Race Games}

Race games are a class of games played on DAGs. For $a, b \in \Nat$, the race game $\race{a,b}$ ends within $a+b$ turns, \PO wins the game if he wins $a$ biddings before \PT wins $b$ biddings. The key property of race games that we employ is that for each vertex $v$ independent of the budgets, there is a neighbor $v_i$ such that \PLi proceeds to $v_i$ upon winning the bidding at $v$, for $i \in \set{1,2}$. 
\cref{fig:race} depicts \race{3,3}.
\begin{figure}[ht]
	\centering
	\begin{tikzpicture}[xscale=0.5, yscale=0.5]

		\draw (0,0) node[rond5, minimum size = 0.75cm] (n33) {\(v_{3,3}\)};
		\draw (2,0) node[rond5] (n23) {\(v_{2,3}\)};
		\draw (4,0) node[rond5] (n13) {\(v_{1,3}\)};

		\draw (0,-2) node[rond5] (n32) {\(v_{3,2}\)};
		\draw (2,-2) node[rond5] (n22) {\(v_{2,2}\)};
		\draw (4,-2) node[rond5] (n12) {\(v_{1,2}\)};

		\draw (0,-4) node[rond5] (n31) {\(v_{3,1}\)};
		\draw (2,-4) node[rond5] (n21) {\(v_{2,1}\)};
		\draw (4,-4) node[rond5] (n11) {\(v_{1,1}\)};

		\draw (6,-2) node[rect,minimum size=0.5cm] (t) {\(t\)};
		
		\draw (2, -6) node[diam,minimum size=0.75cm] (b) {\(s\)};
		
		\draw (n33) edge[->] (n23);
		\draw (n23) edge[->] (n13);
		\draw (n13) edge[->] (t);
		
		\draw (n33) edge[->] (n32);
		\draw (n23) edge[->] (n22);
		\draw (n13) edge[->] (n12);

		\draw (n32) edge[->] (n22);
		\draw (n22) edge[->] (n12);
		\draw (n12) edge[->] (t);

		\draw (n32) edge[->] (n31);
		\draw (n22) edge[->] (n21);
		\draw (n12) edge[->] (n11);
		
		\draw (n31) edge[->] (b);
		\draw (n21) edge[->] (b);
		\draw (n11) edge[->] (b);

		\draw (n31) edge[->] (n21);
		\draw (n21) edge[->] (n11);
		\draw (n11) edge[->] (t);
		
	\end{tikzpicture}
	\caption{\(\race{3,3}\)}\label{fig:race}
\end{figure}

In the following, we establish closed-form of threshold budgets at any vertex of a race game \race{a, b} by induction. 

\begin{restatable}{theorem}{thresholdrace}
\label{thm:thresholdrace}
Let $v$ be the root of a race game $\race{a,b}$. Then $T_v(B_2) = a\cdot \floor{B_2/b}$.
\end{restatable}

\begin{proof}
	First note that, the threshold budget for \PO is \(0\) at \(t\), and \(\infty\) at \(s\).
	Let us denote any vertex of the race game as \(v_{x,y}\) where \(x\) and \(y\) are the minimum distance from the vertex to \(t\) and \(s\),  respectively.
	In this notation, the root \(v\) of \(\race{a,b}\) is referred to as \(v_{a,b}\).
	
	Note that, a subgame of \(\race{a,b}\) rooted at any vertex \(v_{x, y}\) is \(\race{x,y}\) itself. 
	We, in fact, show in the following: \(T_{v_{x, y}} = x \cdot \floor{\frac{B}{y}}\) which implies what we require.
	
	Let us now consider \(v_{1,1}\).
	At this vertex, \PO has to win the bid, otherwise \PT simply moves the token to \(s\).
	Because \PT has a budget of \(B\), and \PO wins all ties,
	his threshold budget at this vertex is \(B\), and he bids his whole budget.
	
	By induction on \(x\), we can argue that for any vertex of the form \(v_{x, 1}\), the threshold budget is \(xB\), because \PO has to win all \(x\) the bids to prevent the token reaching \(s\).
	Thus he has to bid at least \(B\) at all those \(x\) bids.
	In fact, if he has budget at most \(xB - 1\) at vertex \(v_{x,1}\), then \PT has a winning strategy:
	she bids \(B\) until she wins.
	
	Similarly, by induction on \(y\), we claim that for any vertex of the form \(v_{1, y}\), the threshold budget is \(\floor{\frac{B}{y}}\).
	The base case of this induction is \(v_{1,1}\), for which we showed earlier that the statement is true.
	Let us assume it is true for \(v_{1,y-1}\), and we prove the claim for \(v_{1,y}\).
	
	We suppose \PO's budget at \(v_{1, y}\) is \(\floor{\frac{B}{y}}\) while \PT's budget is \(B\).
	We claim that the winning strategy for \PO at vertex \(v_{1,y}\) is to bid his whole budget itself.
	If he wins the bid, he moves the token to target.
	Otherwise, \PT wins the bid by at least bidding \(\floor{\frac{B}{y}} +1\), and of course, she moves the token to \(v_{1, y-1}\) .
	Thus, her budget at \(v_{1, y-1}\) is at most \(B - (\floor{\frac{B}{y}} +1)\), while \PO's budget remains \(\floor{\frac{B}{y}}\).
	From the induction hypothesis, we know when \PT has a budget \(B - (\floor{\frac{B}{y}} +1)\), \PO's threshold budget for surely winning is \(\floor{\frac{B - (\floor{\frac{B}{y}} +1)}{y-1}}\).
	If we can show that, \(\floor{\frac{B}{y}} \geq \floor{\frac{B - (\floor{\frac{B}{y}} +1)}{y-1}}\), we are done.
	We show this in the following:
	\begin{align*}
		B - y \cdot \floor{\frac{B}{y}} &\leq y-1\\
		\implies B - \floor{\frac{B}{y}} &\leq (y-1) \cdot \floor{\frac{B}{y}} + (y-1)\\
		\implies \frac{B - \floor{\frac{B}{y}}}{y-1} &\leq \floor{\frac{B}{y}} + 1\\
		\implies \frac{B - (\floor{\frac{B}{y}}+1)}{y-1} &\leq \floor{\frac{B}{y}} + \frac{y-2}{y-1}
	\end{align*}
	
	Because \(\floor{\frac{B}{y}}\) itself is an integer, by taking \(\floor{}\) on the both side, we get \(\floor{\frac{B}{y}} \geq \floor{\frac{B - (\floor{\frac{B}{y}} +1)}{y-1}}\).
	Therefore, \(\floor{\frac{B}{y}}\) is a sufficient budget for \PO to surely win from vertex \(v_{1,y}\).
	
	Now, we need to show that this is also necessary budget for him.
	In fact, we show that when \PO has budget at most \(\floor{\frac{B}{y}} - 1\), while \PT's budget is \(B\), she has a surely winning strategy from \(v_{1,y}\).
	Her winning strategy is bidding \(\floor{\frac{B}{y}}\), until she reaches \(s\).
	Because \(B \geq y \cdot \floor{\frac{B}{y}}\), she can actually bids likewise.
	At each vertex, \PO's budget will be strictly less than what she is bidding, therefore he looses all the \(y\) bids, and the token indeed reaches the safety vertex.
	
	For a general vertex \(v_{x, y}\), we argue by induction which goes like above.
	We assume that for \(v_{x-1,y}\) and \(v_{x, y-1}\), which are the only two neighbours of \(v_{x, y}\), the threshold budget for \PO is \((x-1) \cdot \floor{\frac{B}{y}}\) and \(x \cdot \floor{\frac{B}{y-1}}\), respectively.
	
	We suppose \PO's budget at \(v_{x, y}\) is \(x \cdot \floor{\frac{B}{y}}\), and \PT's budget is \(B\).
	We claim that his wining strategy at the first bid is to bid \(\floor{\frac{B}{y}}\).
	We show that irrespective of where the token gets placed at the next vertex, he will have the respective threshold budget at that vertex.
	
	If he wins the bid at \(v_{x, y}\), his new budget becomes \((x-1)\cdot \floor{\frac{B}{y}}\), which is exactly what he needs to surely win from \(v_{x-1, y}\).
	If he looses, and the token gets placed at \(v_{x, y-1}\), \PT's budget becomes at most \(B - \floor{\frac{B}{y}}+1)\).
	It remains to show that \(x \cdot \floor{\frac{B -  \floor{\frac{B}{y}}+1)}{y-1}} \leq x\cdot \floor{\frac{B}{y}}\), which is true as we have earlier established \(\floor{\frac{B}{y}} \geq \floor{\frac{B - (\floor{\frac{B}{y}} +1)}{y-1}}\).
	It proves that \(x\cdot \floor{\frac{B}{y}}\) is the sufficient budget for \PO to surely win from \(v_{x, y}\).
	
	Finally, if \PO's budget is at most \(x \cdot \floor{\frac{B}{y}} - 1\) and \PT's budget is \(B\), then \PT wins the game if she bids \(\floor{\frac{B}{y}}\) at each bidding.
	This can be shown by another inductive argument where we assume the statement being true for vertices \(v_{x-1, y}\) and \(v_{x, y-1}\), and follow the same steps that we did for \PO above.
\end{proof}

With the exact closed-form of threshold budgets for race games, we now show that the bounds in \cref{stm:pipe} are tight.

\begin{corollary}  \label{stm:pipe_tight}
	For every rational number $q = n / m$, there exist infinitely many games $\G$ with vertex $v$ such that $t_v = q$ and for infinitely many $B$ the lower and upper bound of~\cref{stm:pipe} actually is an equality for some $B_2 > B$.
\end{corollary}
\begin{proof}
	Choose $\G = \race{n, m}$ (or any multiple thereof) and insert the closed form of \cref{thm:thresholdrace}.
	Note that in a race game $\textnormal{max-path}(v)$ of the root vertex $v$ clearly is $\max(n,m)$.
\end{proof}


\subsection{Tug-of-War games}
Given an integer \(n \geq 1\), a \textit{tug-of-war} game $\TUG(n)$ is a game played on a chain with $n+2$ nodes, namely $n$ interior nodes and two endpoints $s$ and $t$. We develop closed-form representations of thresholds in  \(\TUG(2)\) and \(\TUG(3)\) (both depicted in \cref{fig:tow}). 
For integers $k\in[1,n]$ and $b\ge 0$, we denote by $\tug(n,k,b)$ the smallest budget that \PO needs to win the tug-of-war game $\TUG(n)$ at the vertex  that is $k$ steps from his target $t$, when the opponent has budget $b$.

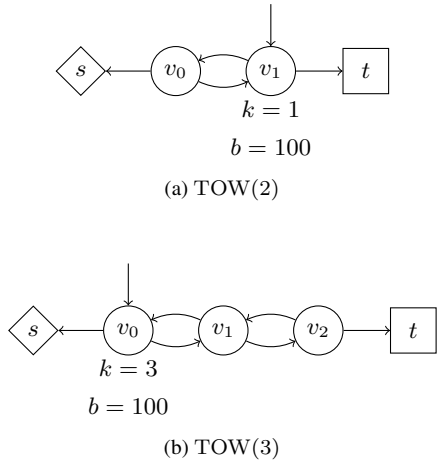
\begin{figure}[ht]
	\captionsetup{justification=centering}
	\begin{subfigure}{0.5\textwidth}
		\centering
		\begin{tikzpicture}[xscale=1.25]
			\draw (0,0) node[state] (0) {\(v_0\)};
			\draw (1,0) node[state] (1) {\(v_1\)};
			\draw (2,0) node[target] (2) {\(t\)};
			\draw (1, 1) node (a) {};
			\draw (1, -.5) node (b) {\(k = 1\)};
			\draw (1, -1) node (c) {\(b = 100\)};
			
			\draw (-1, 0) node[sink] (3) {\(s\)};
			
			\draw (a) edge[->] (1);
			\draw (0) edge[->, bend right = 30] (1);
			\draw (1) edge[->, bend right = 30] (0);
			\draw (1) edge[->] (2);
			\draw (0) edge[->] (3);
			
		\end{tikzpicture}
		\caption{\(\TUG(2)\)}
	\end{subfigure}\vspace{5mm}
	\begin{subfigure}{0.5\textwidth}
		\centering
		\begin{tikzpicture}[xscale=1.25]
			
			\draw (0,0) node[state] (0) {\(v_0\)};
			\draw (1,0) node[state] (1) {\(v_1\)};
			\draw (2,0) node[state] (2) {\(v_2\)};
			\draw (3,0) node[target] (4) {\(t\)};
			
			\draw (-1, 0) node[sink] (3) {\(s\)};
			
			\draw (0, 1) node (a) {};
			\draw (0, -.5) node (b) {\(k = 3\)};
			\draw (0, -1) node (c) {\(b = 100\)};

			\draw (a) edge[->] (0);
			\draw (0) edge[->, bend right = 30] (1);
			\draw (1) edge[->, bend right = 30] (0);
			\draw (1) edge[->, bend right = 30] (2);
			\draw (2) edge[->, bend right = 30] (1);
			\draw (2) edge[->] (4);
			\draw (0) edge[->] (3);
		\end{tikzpicture}
		\caption{\(\TUG(3)\)}
	\end{subfigure}
	\caption{Examples of tug-of-war games for \(n = 2 \text{ and } 3\) respectively}\label{fig:tow}
\end{figure}

\begin{restatable}{theorem}{tugtwo}\label{thm:tug2}
	For \(b \geq 0\), we have 
        \(\tug(2,1,b)=\floor{b/\phi}\) and \(\tug(2,2,b)=\floor{b \cdot \phi}\),
       where \(\phi = (\sqrt{5}+1)/2 \approx 1.618\) is the golden ratio.
\end{restatable}

\begin{proof}
	To simplify the notation, let us assume \(t_b = \tug(2,1,b)\), and \(u_b = \tug(2,2,b)\).
	We first claim that \(t_b\) and \(u_b\) are the unique solution to the following system of recurrence relations. 
	\begin{enumerate}
		\item\label{itm:tuga} \(t_0 = u_0 = 0\)
		
		\item\label{itm:tugb} \(u_b = t_b + b\) for any \(b \geq 1\)
		
		\item\label{itm:tugc} \(t_b = \min_x\{\max (x, u_{b - 1-x}) \mid 0 \leq x \leq b\}\) for any \(b \geq 1\)
	\end{enumerate}
	\cref{itm:tuga} is obvious because \PO bids \(0\) at every step and he wins ties, when \PT has a budget \(0\).
	
	\PO needs to win at the vertex which is \(2\) steps away from his target, otherwise \PT moves the token to the other end-point.
	Therefore, \PO needs to bid \(b\), and his new budget should be, by definition, at least \(t_b\) upon winning.
	This gives us \cref{itm:tugb}.
	
	Finally, at the vertex which is a single step away from \PO's target, he needs to optimize what his bid would be between \(0\) and \(b\) so that even if he loses the current bid, he would have enough budget at the next step to win from there (i.e, \(u_b\)).
	This gives us \cref{itm:tugc}.
	
	Moreover, the system of equations has a unique solutions, as there are as many equations as there are unknowns (\(t_b, u_b\) for a fixed \(b\)).
	Hence, it is enough to show that the expressions \(t_b = \floor{\frac{b}{\phi}}\) and \(u_b = \floor{b \cdot \phi}\) satisfy those equations.
	Clearly, \(\floor{\frac{0}{\phi}} = \floor{0 \cdot \phi} = 0\), so \cref{itm:tuga} holds.
	Next note that the golden ratio satisfy \(\phi = 1 + 1/\phi\). 
	Thus,
	\[ u_b = \floor{b \cdot \phi} = \floor{b \cdot (1 + 1/\phi)} = \floor{b + b/\phi} = b + \floor{b/\phi} = b + t_b\]
	
	implying \cref{itm:tugb} holds too.
	
	Finally, note that the function $f\colon x\to x$ is increasing, hence to verify \cref{itm:tugc} we need to show two inequalities for any $b\geq 1$:
	\begin{enumerate}
		\item For $x=\floor{b/\phi}$ we have $\floor{(b-1-x)\cdot \phi}\leq \floor{b/\phi}$.
		\item For $x=\floor{b/\phi}-1$ we have $\floor{(b-1-x)\cdot \phi}\geq \floor{b/\phi}$.
	\end{enumerate}
	In both cases, we will do this by checking that the insides of the two floor functions being compared satisfy the same inequality. Upon plugging in $x$, it thus suffices to show
	\[(b-1-\floor{b/\phi})\cdot \phi \leq b/\phi
	\quad\text{and}\quad
	(b-\floor{b/\phi})\cdot \phi \geq b/\phi.
	\]
	From $\phi=1+1/\phi$ we have $b\cdot\phi-b/\phi=b$, so the desired inequalities rewrite as
	\[ b-\phi \leq \floor{b/\phi}\cdot\phi
	\quad\text{and}\quad
	\floor{b/\phi}\cdot\phi\leq b.
	\]
	Those two inequalities follow from the obvious inequalities $b/\phi-1\leq \floor{b/\phi} \leq b/\phi$ after multiplying by $\phi$.
\end{proof}

\begin{remark}
\label{rem:Wythoff}
\normalfont
The closed-form solution in \cref{thm:tug2} has a striking similarity to a classic result in Combinatorial Game Theory. {\em Wythoff Nim} is played by two players who alternate turns in removing chips from two stacks. A configuration of the game is $\zug{s_1, s_2}$, for integers $s_1 \geq s_2 \geq 0$, representing the number of chips placed on each stack. A player has two types of actions: (1) choose a stack and remove any $k>0$ chips from that stack, i.e., proceed to $\zug{s_1 -k, s_2}$ or $\zug{s_1, s_2 - k}$, or (2) remove any $k >0$ chips from both stacks, i.e., proceed to $\zug{s_1 - k, s_2 -k}$. The player who cannot move loses. Wythoff~\cite{Wyt07} identified the configurations from which the first player to move loses. Trivially, $\zug{0,0}$ is losing, followed by $\zug{1, 2}, \zug{3,5}, \ldots$. In general, the $n$-th losing configuration is $\zug{\lfloor n \cdot \phi \rfloor, \lfloor n \cdot \phi \rfloor + n}$. Note the similarity to the thresholds in  $v_2$ and $v_1$, which can be written respectively as $\zug{\lfloor b \cdot \phi \rfloor, \lfloor b \cdot \phi \rfloor - b}$, for $b \geq 0$.
\end{remark}

\begin{restatable}{theorem}{tugthree}\label{thm:tug3}
	For $b\ge 1$ we have
	\(\tug(3,1,b)=\floor{\frac{b-1}{2}}\), \(\tug(3,2,b)=b-1\), and \(\tug(3,3,b)=2b-1\).
\end{restatable}

\begin{proof}
	We proceed similarly to the proof of~\cref{thm:tug2}.
	This time, we need to check that the expressions
	\[t_b=\floor{(b-1)/2}, \quad u_b=b-1, \quad\text{and}\quad v_b=2b-1
	\]
	satisfy the relations
	\begin{enumerate}
		\item\label{itm:tug3a} $t_1=u_1=0$, $v_1=1$,
		\item\label{itm:tug3b} $v_b=u_b+b$ for any $b\geq 2$,
		\item\label{itm:tug3c} $u_b= \min_x \{  \max\{t_b+x, v_{b-1-x}\} \mid 0\leq x \leq b\}$ for any $b\geq 2$.
		\item\label{itm:tug3d} $t_b= \min_x \{  \max\{x, u_{b-1-x}\} \mid 0\leq x \leq b\}$ for any $b\geq 2$.
	\end{enumerate}
	This time, both \cref{itm:tug3a} and \cref{itm:tug3b} follow by direct substitution.
	
	Regarding \cref{itm:tug3c}, we need to show that
	\[ b-1=  \min_x \{  \max\{\floor{(b-1)/2}+x, 2b-3-2x \} \mid 0\leq x \leq b\}
	\]
	To that end, we distinguish two cases based on the parity of $b$.
	If $b=2k$ is even then we need to show
	\[ 2k-1 = \min_x\{ \max\{k-1+x, 4k-3-2x \} \mid 0\leq x \leq 2k\},
	\]
	and indeed the minimum on the right-hand side is attained for $x=k-1$ and is equal to $2k-1$ as desired.
	Similarly, if $b=2k+1$ is odd then we need to show
	\[ 2k = \min_x\{ \max\{k+x, 4k-1-2x \} \mid 0\leq x \leq 2k\},
	\]
	and indeed the minimum on the right-hand side is attained for $x=k$ and is equal to $2k$ as desired.
	
	Finally, regarding \cref{itm:tug3d} we have
	$u_{b-1-x}=b-2-x$, hence the two numbers inside the $\max(\cdot)$ function always sum up to $b-2$.
	If $b=2k$ is even, then the minimum is $(b-2)/2=k-1 = \floor{(b-1)/2}=t_b$ as desired.
	If $b=2k+1$ is odd then the minimum is $\ceil{(b-2)/2}=k=\floor{(b-1)/2}=t_b$ as desired again.
\end{proof}

We note that for $n \ge 4$ the situation gets surprisingly more complicated.
For $n=5$ the threshold budgets do eventually converge to a simple pattern, but only from around $b=4\cdot 10^3$ on.
In contrast, for $n\in\{4,6\}$ the threshold budgets exhibit no clear pattern up until $b=10^6$.
Moreover, while the pipe theorem \cref{stm:pipe} seems to hold for $n \leq 5$ (experimentally validated up to $b = 10^7$), it is (quickly) violated for $n \geq 6$.
This suggests that a simple closed form solution for general games is unlikely, given that these structurally similar games behave so differently.

\section{Algorithms for Threshold Budgets}\label{sec:thresholdgeneral}
In this section, we discuss an algorithmic approach to compute threshold budgets. We point out that the Pipe theorem (\cref{stm:pipe}) only provides an approximation for the thresholds, and periodicity (\cref{thm:periodicity}) only holds eventually, thus, in order to use it, exact thresholds need to be computed until periodicity ``kicks in''.
We study the following problem: Given a game $\G$, a vertex $v$ in $\G$, and a budget $B_2$ of \PT, determine $T_v(B_2)$. We develop an algorithm for general games, running in time pseudo-polynomial in $B_2$ and polynomial in $|\G|$, and then a specialized variant for DAGs which is pseudo-\emph{linear} in $B_2$.
In the following, we write $B$ for an ``arbitrary'' \PT budget and $B_2$ for the particular budget for which we want to compute $\budget{B_2}{v}$.

As a first step, we show that poorman discrete-bidding games end after a finite number of steps.
Consider a vertex $v$. We define the \emph{maximal step count}, denoted $\gamestepbound{\G}{B}$, to be the maximal number of steps \PT can delay reaching $t$ when the initial budgets are $B$ and $\budget{B}{v}$ for \PT and \PO, respectively, and \PO follows some winning strategy.
Let $\gamestepbound{\G}{B} = \max_v \budget{B}{v}$. The following lemma bounds $\gamestepbound{\G}{B}$.
\begin{lemma} \label{stm:bounded_game}
	Given a budget of $\budget{B_2}{v}$, \PO can ensure winning after at most $\mathcal{O}(|V| \cdot B_2)$ steps.
\end{lemma}
\begin{proof}
	If \PT does not win a bid for $|V|$ steps, then \PO can surely move to the target $t$.
	Otherwise, \PT has to win at least one bid, decreasing the budget by at least 1 every $|V|$ steps. 
\end{proof}
We note that this is a very crude approximation, we conjecture that actually $\gamestepbound{\G}{B} \in \mathcal{O}(\log B)$, as we explain later.
However, the existence of such a bound already motivates us to consider the step-bounded variant of the game:
Let \(\stepbudget{B}{i}{v}\) equal the minimal budget that \PO needs to ensure winning from \(v\) against a budget of \(B\) \emph{in at most \(i\) steps} (or $\infty$ if this is not possible).
By Lem.~\ref{stm:bounded_game}, \(\stepbudget{B}{i}{v} = \budget{B}{v}\) for some large enough $i$.
Thus, we are interested in computing $\stepbudget{B}{i}{v}$ for increasing $i$ until convergence.
Let us briefly discuss simple cases.
For the target vertex, clearly \(\budget{B}{t} = \stepbudget{B}{i}{t} = 0\) for any \PT budget \(B\) and any \(i\).
For the sink, \(\budget{B}{s} = \stepbudget{B}{i}{s} = \infty\), as well as $\stepbudget{B}{0}{v} = \infty$ for all non-target vertices.
As it turns out, we can compute all other values by a dynamic programming approach.

We first describe a recursive characterization of $\stepbudget{B}{i}{v}$, which then immediately yields our algorithm.
To this end, we consider the \emph{step operator} $\stepop{v}{f}{b}{B}$, which given a threshold function $f$ (such as $\stepbudget{B}{i}{v}$) and vertex $v$ yields the outcome of placing bid $b$ as \PO against a \PT budget $B$.
The intuition is as follows:
Suppose $f$ is the actual threshold required to win in every vertex.
There are two distinct cases.
If \PO bids $B$, i.e.\ all of \PT's budget, a win of the auction is guaranteed.
\PO pays $B$ and then naturally moves to the ``cheapest'' successor, i.e.\ one with minimal threshold as given by $f$.
Otherwise, with a bid of $b < B$ by \PO, \PT could either bid $0$, again leaving \PO to pay $b$ and choose the best option, or bid $b + 1$, i.e.\ \PT wins instead, paying the bid and choosing the most expensive successor.
The overall best choice for \PO then directly is given as minimum over all sensible bids.
\begin{definition}\label{def:step_function}
	Let $B$ a budget for \PT and a function \(f : V \times \{0, \dots, B\} \to \bbN\) yielding a threshold for each budget (e.g. $\stepbudget{B}{i}{v}$).
	We define $\stepop{v}{f}{B}{B} = B + \min_{v' \in N(v)} f(v', B)$ and, for any other bid $0 \leq b < B$, let
	\begin{equation*}
		\stepop{v}{f}{b}{B} = \max \begin{cases}
			b + \min_{v' \in N(v)} f(v', B)     \\
			\max_{v' \in N(v)} f(v', B - (b+1))
		\end{cases}
	\end{equation*}
	Finally, $\step{v}{f}{B} = {\min}_{0 \leq b \leq B} \stepop{v}{f}{b}{B}$.
\end{definition}
Indeed, $\mathrm{step}$ allows us to iteratively compute $T_v^i$ as follows:
\begin{lemma} \label{stm:recursive_equation}
	For all $i > 0$, we have $\stepbudget{B}{i}{v} = \step{v}{T_{\circ}^{i-1}}{B}$.
\end{lemma}
\begin{proof}
	We proceed by induction over $i$.
	The correctness of the base cases follows immediately.
	To go from step $i - 1$ to $i$, observe that \PO surely never wants to bid more than $B$, since this bid suffices to guarantee winning.
	Moreover, for any fixed bid $b < B$, the opponent \PT either wants to bid $0$, letting \PO win, or $b + 1$, claiming the win at minimal potential cost:
	Bidding anything between $0$ and $b$ as \PT does not change the outcome, and bidding more than $b + 1$ certainly is wasteful.
	By this observation, we can immediately see that for each potential bid $b$ between $0$ and $B$, $\stepop{v}{T_{\circ}^{i-1}}{b}{B}$ yields the best possible outcome against an optimal opponent.
	In particular, if \PO bids $b$ but the available budget is one smaller than $\stepop{v}{T_{\circ}^{i-1}}{b}{B}$, then there exists a response of \PT where \PO is left with less budget than $\stepbudget{B'}{i - 1}{v'}$ in some vertex $v'$ against \PT budget $B'$, which by induction hypothesis is not sufficient.
\end{proof}
%
%
%
%
%
%
%

\begin{algorithm}[t]
	\caption{Iterative Algorithm to compute threshold budgets} \label{algo}
	\begin{algorithmic}
		\Require Game $\G = (V, E, t, s)$, \PT{} budget $B_2$
		\Ensure Thresholds for every $v \in V$ and $0 \leq B \leq B_2$
		\State Set $f_i(t, \circ) \gets 0$ and $f_i(s, \circ) \gets \infty$ for all $i \geq 0$
		\State Set $f_0(v, \circ) \gets \infty$ for all $v \notin \{t, s\}$
		\State Set $i \gets 0$
		\While{$f_i$ changes in the iteration}
			\For{$v \in V \setminus \{t, s\}$, $0 \leq B \leq B_2$}
				\State $f_{i+1}(v, B) = \step{v}{f_i}{B}$
			\EndFor
			\State $i \gets i + 1$
		\EndWhile
		\State \Return $f_i$
	\end{algorithmic}
\end{algorithm}

This naturally gives rise to an iterative algorithm:
Given budget $B_2$, we compute $\stepbudget{B}{i}{v}$ for all vertices $v$ and budgets $0 \leq B \leq B_2$ for increasing $i$ until a fixpoint is reached.
We briefly outline the algorithm in \cref{algo}.

At first glance, evaluating $\step{v}{f}{B}$ requires $\mathcal{O}(B \cdot |N(v)|)$ time -- we need to consider all possible bids and go over all successors.
Thus, to compute $\stepbudget{B}{i}{v}$ for all $B \leq B_2$ and vertices $v$ takes $\mathcal{O}(B_2^2 \cdot |E|)$.
(By our assumption, every vertex has at least one outgoing edge, meaning $|V| \in \mathcal{O}(|E|)$.)
While the graphs (and thus $|E|$) we consider typically are small, quadratic dependence on $B_2$ is undesirable, since we may want to compute optimal solutions for considerably large budgets.
It turns out that we can exploit some properties of $\stepbudget{B}{i}{v}$ to obtain speed-ups.
\begin{theorem} \label{stm:complexity}
	For budget $B_2$ of \PT, the threshold budget can be determined in $\mathcal{O}(\gamestepbound{\G}{B_2} \cdot B_2 \cdot \log(B_2) \cdot |E|)$.
\end{theorem}
\begin{proof}
	Observe that $\stepbudget{B}{i}{v}$ is monotone in $B$:
	Winning against a larger budget of \PT certainly requires the same or more resources.
	Thus, the first expression of the maximum in ~\cref{def:step_function} is a (strictly) monotonically increasing function, while the second is decreasing.
	Together, the step function intuitively is convex in $b$:
	There is a ``sweet spot'', bidding too much is not worth it and bidding too little lets \PT gain too much.
	Consequently, we can determine $\stepbudget{B}{i}{v}$ by a binary search between $0$ and $B$.
	This yields a running time of $\mathcal{O}(\log B \cdot |N(v)|)$ for a fixed vertex $v$ and budget $B$.
	In turn, to compute a complete step, i.e.\ for all vertices determine $\stepbudget{B}{i}{v}$ for all budgets $B \leq B_2$, we get $\mathcal{O}(B_2 \cdot \log(B_2) \cdot |E|)$.
	(Note that $\sum_{v \in V} |N(v)| = |V|$.)
\end{proof}
%
%
\subsection{A Pseudo-Linear Algorithm for DAGs}
Using insights of the previous section together with further observations, we can obtain tighter bounds in the case of DAGs.
In particular, by exploiting both the given topological ordering as well as the bounds given by Thm.~\ref{stm:pipe}, we obtain an algorithm linear in the numerical value of $B_2$.
\begin{restatable}{theorem}{lineardagalgo} \label{stm:linear_dag_algo}
	For a DAG game and any budget $B_2$ of \PT, the threshold budget $\budget{B_2}{v}$ can be determined in $\mathcal{O}(B_2 \cdot \log(|V|) \cdot |E|)$ steps for all vertices.
\end{restatable}

\begin{proof}
	Fix the input as in the assumptions.
	
	Firstly, we see that each vertex of the DAG is evaluated exactly once and we can, in one step, directly compute $\budget{B}{v}$ for all $0 \leq B \leq B_2$ one vertex at a time:
	Sort the vertices in reverse topological order.
	Observe that, by assumption, sink and target are the only leaves, for which computing $\budget{B}{v}$ is trivial.
	Then, inductively, whenever we compute $\budget{B}{v}$, the values $\budget{\cdot}{v'}$ of all successors $v' \in N(v)$ are already known.
	Thus, we can directly compute $\budget{B}{v} = \step{T_\circ}{v}{B}$.
	(Note that this reasoning also can be applied to the SCC decomposition of general games.)
	
	Secondly, using Thm.~\ref{stm:pipe}, we can derive bounds on the optimal bid:
	We know that the threshold $\budget{B}{v}$ in a particular state $v$ has to lie between the lower and upper bounds given by the theorem -- a linearly sized interval.
	This however does not immediately give us bounds on the bids.
	Using the above approach of processing vertices in reverse topological order, whenever we handle a given vertex $v$, all of its successors are already solved.
	Together, we know (i)~a linearly sized interval of potential \emph{thresholds} for $v$, say $[B^-, B^+]$ and (ii)~the exact thresholds in all successor vertices.
	Note that in order to use Thm.~\ref{stm:pipe} computationally, we first need to determine the continuous ratios $t_v$ for every vertex.
	We explain afterwards how this can be achieved in linear time, too.
	
	We define $T' = \min_{v' \in N(v)} \budget{B}{v'}$ the smallest threshold over all successors against $B$, i.e.\ the minimum budget \PO needs to win after winning the bid in $v$ (and paying for it).
	As an immediate observation, we see that an optimal bid can never be larger than $B^+ - T'$:
	If \PO would bid more than $B^+ - T'$, \PT bids 0 in response, leaving \PO with a budget of less than $T'$, which is required to win.
	
	For the lower bound, we prove that at least one optimal bid is at least $B^- - T'$.
	(This does not exclude optimal bids which are smaller than $B^- - T'$.)
	Suppose that the threshold budget is $\budget{B}{v'} = B_1$ and there is a winning strategy for \PO with a bid $b < B^- - T'$.
	We consider the bid $b' = B^- - T' > b$.
	If \PO wins with $b'$, a budget of $B_1 - b' = (B_1 - B^-) + T')$ is left, which is at least $T'$, since $B_1 \geq B^-$ by assumption.
	By definition, \PO can pick a successor from which a winning strategy with budget $T'$ or larger exists.
	For the losing case, recall that the bid $b$ was winning.
	This means that \PO can win if \PT wins by bidding $b + 1$.
	In particular, in every successor of $v$ a budget of $B_1$ is sufficient to win against $B - (b + 1)$ (which is \PT's budget afterwards).
	Thus, if \PO instead bids $b'$ and \PT wins (by bidding $b' + 1$), observe that $B - (b' + 1) < B - (b + 1)$, since $b + 1 > b' + 1$ -- \PT is left with even less budget than before.
	
	Together, we know that an optimal bid exists in the interval $[B^- - T', B^+ - T']$.
	Thus, we can restrict ourselves to checking all possible bids in this interval.
	Observe that $B^+ - B^-$ is linear in the size of the graph by Thm.~\ref{stm:pipe}, in particular it is bounded by the number of vertices times the largest continuous ratio.
	Moreover, we can apply the binary search idea of Thm.~\ref{stm:complexity}.
	In summary, we obtain a complexity of $\mathcal{O}(\log(B^+ - B^-) \cdot N(v))$ to determine $\budget{B}{v}$ for a vertex $v$ and budget $0 \leq B \leq B_2$.
	
	It remains to prove complexity and size bounds on $t_v$.
	First, we observe that given the ratios $t_{v'}$ of all successors, we can immediately compute $t_v = t_v^+ / (1 + t_v^+) \cdot (1 + t_v^-)$, where $t_v^+ = \max_{v' \in N(v)} t_{v'}$ and $t_v^- = \min_{v' \in N(v)} t_{v'}$ (using the results of, e.g., \cite[Sec.~3]{LLPSU99}).
	Note that $t_v^+ = \infty$ if $s \in N(v)$.
	In that case, we have $t_v = 1 + t_v^-$.
	As such, we can again obtain all ratios by a linear pass in reverse topological order.
	Moreover, the bit size of $t_v$ is bounded by the sum of bit sizes of $t_v^-$ and $t_v^+$, i.e.\ $|t_v|_\# \in \mathcal{O}(|t_v^+|_\# + |t_v^-|_\#)$, where $|v|_\#$ denotes the size of the representation of $v$.
	Since the ratio of sink and target are trivial (i.e.\ of size $1$), we obtain, as a crude upper bound, $|t_v|_\# \in \mathcal{O}(|V|^2)$ for all $v$.
	This means that evaluating the equation takes at most $\mathcal{O}(|V|^2 \log |V|)$ time (note that $|N(v)| \leq |V|^2$) and we can obtain $t_v$ for all vertices in time $\mathcal{O}(|V|^3 \log |V|)$.
	
	We also directly obtain a bound on the magnitude of $t_v$:
	Clearly, $t_v \leq 1 + t_v^-$, i.e.\ $t_v \in \mathcal{O}(|V|)$.
	Also, this bound is tight:
	In $\race{1,n}$, \PO needs $|V| - 2$ times the budget of \PT, since its required to win $|V|$ in a row without alternative.
	Consequently, the ``height'' of the pipe, i.e.\ $B^+ - B^-$ is at most of size $|V|^2$.
	
	Combining all results, we obtain that the overall complexity of this algorithm is bounded by
	\begin{equation*}
		\mathcal{O}(|V|^3 \log(|V|) + B_2 \cdot \log(|V|) \cdot |E|).
	\end{equation*}
	Note that if $B_2 \leq |V|^2$ we can employ our ``classical'' algorithm which simply applies binary search from $0$ to $B_2$ in reverse topological order, yielding a complexity of $\mathcal{O}(B_2 \cdot \log(B_2) \cdot |E|)$ (requiring for each vertex $\mathcal{O}(\log(B_2) \cdot N(v))$).
	Otherwise, i.e.\ $B_2 \geq |V|^2$, $B_2 \cdot \log(|V|) \cdot |E|$ dominates $|V|^3 \log(|V|)$ (recall that $|V| \leq |E|$), proving the claim.
\end{proof}
\stam{
\begin{proof}[Proof sketch](See supplementary material for the full proof.)

	In essence, we use four observations.
	First, since the game is a DAG, we can fully compute $\budget{B}{v}$ for each vertex and budget $0 \leq B \leq B_2$ at once by evaluating vertices in reverse topological order.
	Intuitively, each vertex only occurs at most once along any play in a DAG game.
	Thus, we only need to consider each vertex once.
	Second, we can exploit the budget bounds given by Thm.~\ref{stm:pipe} to obtain lower and upper bounds on an optimal bid.
	The size of this interval as given by Thm.~\ref{stm:pipe} depends on the magnitude of the continuous thresholds.
	Thirdly, we rely on actually knowing these thresholds.
	Thus, we give a bound on the size and computational complexity of determining them.
	Finally, applying binary search to this interval, using the insights of Thm.~\ref{stm:complexity}, yields the result.
\end{proof}}

\section{Experiments and Conjectures}
In this section, we present several experimental results which in turn motivate conjectures for general games.

\subsection*{The Pipe Theorem}
In our experiments, we observed that Thm.~\ref{stm:pipe} does not hold for all general graphs.
We depict the smallest bidding game we found where Thm.~\ref{stm:pipe} is violated in Fig.~\ref{fig:pipe_violated_general}.
We note that this game has an interesting structure:
It is a ``normal'' tug of war game, with a single edge added.
Moreover, whenever this ``gadget'' is a part of a game, the same problem arises.
However, this structure is not the only potential cause:
While the pipe theorem even seems to hold for tug of war games of up to 5 interior states (validated up to $B_2 = 10^7$), we observed that it is violated for 6 or more.

\begin{figure}[t]
	\centering
	\begin{tikzpicture}[auto,xscale=1.25]
		\draw node[sink] at (0,0) (0) {0};
		\draw node[state] at (1,0) (1) {1};
		\draw node[state] at (2,0) (2) {2};
		\draw node[state] at (3,0) (3) {3};
		\draw node[state] at (4,0) (4) {4};
		\draw node[target] at (5,0) (5) {5};
		
		\path[->]
			(1) edge (0)
			(1) edge[bend left] (2)
			(2) edge[bend left] (1)
			(1) edge[bend left=55] (3)
			(2) edge[bend left] (3)
			(3) edge[bend left] (2)
			(3) edge[bend left] (4)
			(4) edge[bend left] (3)
			(4) edge (5)
		;
	\end{tikzpicture}
	\caption{
		A small game where Thm.~\ref{stm:pipe} is violated.
	} \label{fig:pipe_violated_general}
\end{figure}
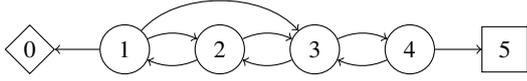

\subsection*{Conjectures on General Graphs}
Despite this apparently chaotic behaviour, we observed that a variant of Thm.~\ref{stm:pipe} seems to be satisfied in general.
\begin{conjecture}
	In any game and vertex $v$, we have that
	\begin{equation*}
		t_v \cdot B_2 - \mathcal{O}(\log B_2) \leq T_v(B_2).
	\end{equation*}
\end{conjecture}
Consider Fig.~\ref{fig:budget_difference}, where we plot the difference $d(B) = t_v \cdot B - T_v(B)$ for a tug-of-war game with 21 states.
The x-axis, i.e.\ \PT's budget $B$, is scaled logarithmically.
If the conjecture holds, then $d(B) \in \mathcal{O}(\log B)$, which would appear as a line on such a graph.
And indeed, we clearly see a linear ``pipe''.
We observed similar graphs for all investigated games.

\begin{figure}[t]
	\centering
	\hspace{0.15\columnwidth}%
	\begin{minipage}[c]{0.7\columnwidth} 
	\begin{tikzpicture}[trim axis left,trim axis right]
	\begin{axis}[
			scale only axis,
			width=\columnwidth,height=2.5cm,
			ylabel=Difference to $t_v \cdot B_2$,
			legend pos=north west,xmode=log,
			xmin=10,xmax=3900,ymax=0
		]
		\addplot+[only marks,draw=black,mark size=0.25pt,mark=x] table [x index=0, y index=1, col sep=comma] {data/bid_difference.csv};
	\end{axis}
	\end{tikzpicture}

	\vspace{1ex}

	\begin{tikzpicture}[trim axis left,trim axis right]
	\begin{axis}[
			scale only axis,
			width=\columnwidth,height=2cm,
			xlabel=\PT's budget $B_2$,ylabel=Optimal bid,
			legend pos=north west,
			xmin=0,xmax=3900,ymin=0,
		]
		\addplot+[only marks,draw=black,mark size=0.25pt,mark=+] table [x index=0, y index=1, col sep=comma] {data/witness.csv};
	\end{axis}
	\end{tikzpicture}
	\end{minipage}
	\caption{
		Plot of $t_v \cdot B_2 - T_v(B_2)$ (logarithmic) and \PO's optimal bids in state 1 of a tug-of-war with 21 states.
	} \label{fig:budget_difference}
\end{figure}
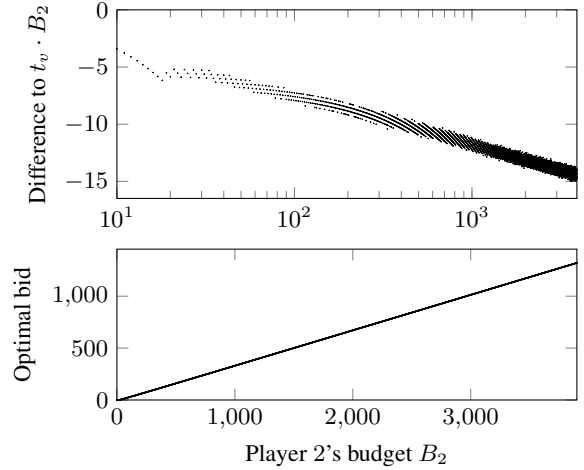

Based on experimental evidence, we believe that the underlying reason is similar to the proof idea of Thm.~\ref{stm:pipe}, namely that for large budgets, the actual bids do not differ too much from the continuous behaviour.
\begin{conjecture}
	Winning bids are proportional to the current budget in play, i.e.\ for each vertex there is a ratio $r_v$ such that all winning bis are $b = r_v \cdot B_2 + \mathcal{O}(1)$.
\end{conjecture}
In \cref{fig:budget_difference} we also display optimal bids for \PO in relation to \PT's budget.
A clear linear dependence with a ratio of approximately $r_v \approx \frac{1}{3}$ is visible.

This implies our ``pipe conjecture'' as follows:
When bids are proportional to the budget, then the total budget in play decays exponentially.
Thus, the length of the game is logarithmic in the available budget, i.e.\ $\gamestepbound{\G}{B_2} \in \mathcal{O}(\log B_2)$ for a fixed game.
Recall that in Thm.~\ref{stm:pipe} we prove the lower bound by arguing that \PT needs a ``$+1$'' at most $|V|$ times to compensate for rounding.
With this general step bound, we can similarly argue that this is required at most logarithmic number of times.
In other words, \PO can exploit the ``rounding advantage'' only logarithmically often.
We also mention that this would then put the complexity of our general algorithm at $\mathcal{O}(B_2 \cdot \log(B_2) \cdot |E|)$.

\subsection*{Implementation and Performance}
We implemented our algorithm in Java (executed with OpenJDK 17) and ran it on consumer hardware (AMD Ryzen 3600).
Generation of games and visualization of results was done using Python scripts.

While not the focus of our evaluation, we observed that our implementation can easily handle large graphs and budgets.
For example, solving a tug of war game with 20 states and $B_2 = 10^6$ took around 1 minute (483 steps).


\section{Conclusion}
We study, for the first time, bidding games that combine poorman with discrete bidding. 
On the negative side, threshold budgets in poorman discrete-bidding games exhibit complex behavior already in simple games, in particular in games with cycles. On the positive side, we identify interesting structure: we prove determinacy, in DAGs, we relate the threshold budgets with continuous ratios, and prove that thresholds are periodic. Additionally, our implementation efficiently computes  exact solutions to non-trivial games. We particularly invite the interested reader to explore bidding games using it, the code will be available on demand.

Our work opens several venues for future work:

Theoretically, 
we left several open problems and conjectures. Beyond that, 
poorman discrete-bidding is more amendable to extensions when compared with poorman continuous-bidding, which quickly becomes technically challenging, or Richman discrete-bidding, which is a rigid mechanism. For example, it is interesting to introduce into the basic model,  multi-players or complex objectives, e.g., that take into account left over budgets~\cite{HDM12}. 

Practically, poorman is more popular than Richman bidding since it coincides with the popular first-price auction and discrete- is more popular than continuous-bidding since most if not all practical applications employ some granularity constraints on bids. It is interesting to develop applications based on these games. For example, to analyze and develop bidding strategies 
in sequential auctions or fair allocation of goods~\cite{BEF22}. 
Further, it is interesting to study {\em mechanism design}: synthesize an arena so that the game has guarantees (e.g.,~\cite{MM+22}).

\acknowledgements{This research was supported in part by ISF grant no. 1679/21, ERC CoG 863818 (FoRM-SMArt) and the European Union’s Horizon 2020 research and innovation programme under the Marie Skłodowska-Curie Grant Agreement No. 665385.}

\bibliography{ga}



\end{document}